\documentclass[12pt]{amsart}
\usepackage{graphicx}

\usepackage{amsmath,amsthm}
\makeatletter
\def\amsbb{\use@mathgroup \M@U \symAMSb}
\makeatother

\usepackage{bbold}

\usepackage[all]{xy}

\usepackage{amssymb}
\usepackage{caption}
\usepackage[latin1]{inputenc}
\usepackage{enumerate}
\usepackage{mathrsfs}  
\usepackage{dsfont}

\usepackage[bbgreekl]{mathbbol}

\usepackage{geometry}
\theoremstyle{plain}

\newtheorem{thm}{Theorem}[section]
\newtheorem{cor}{Corollary}[section]
\newtheorem{lemma}{Lemma}[section]
\newtheorem{prop}{Proposition}[section]
\newtheorem{defi}{Definition}[section]

\newtheorem{obs}{Remark}[section]

\numberwithin{equation}{section}

\def\C{\mathcal {C}}

\def\U{\mathcal {U}}

\def\R{\amsbb{R}}

\def\s{\mathcal{S}}
\def\ttau{{\scriptstyle{\mathcal{T}}}}

\begin{document}

\title[A mathematical objection to several relativistic particles]
 {A mathematical objection to the existence of relativistic mechanical systems of several particles}

\author{ J.  Mu\~{n}oz-D{\'\i}az and R. J.  Alonso-Blanco}

\address{Departamento de Matem\'{a}ticas, Universidad de Salamanca, Plaza de la Merced 1-4, E-37008 Salamanca,  Spain.}
\email{clint@usal.es, ricardo@usal.es}

\begin{abstract}
We will prove that, in general, a system formed by several particles moving along relativistic trajectories can not be described by a mechanical system. The contradiction that leads to the previous assertion is due to the fact that a mechanical system defines a second order differential  equation and this, in turn, induces an absolute time that will generally be incompatible with the proper times of the different particles.
\end{abstract}
\bigskip

\maketitle

\setcounter{tocdepth}{1}

\tableofcontents

\centerline{---------------}
\bigskip

\section{Introduction}\label{introduction}

In Classical Mechanics 
it is well known how to describe a system in which several point particles undergo certain given forces or, perhaps, forces of interaction between them. The system thus defined has a number of degrees of freedom according to the number of particles involved. For example,  if we have $N$ particles with 3 degrees of freedom each, the common mechanical system is described by a configuration space of dimension $3N$.

On the other hand, in the theory of Special Relativity, the configuration space of a single particle is four-dimensional and includes a time coordinate (Minkowski spacetime). Therefore, if it is required to describe a system consisting of several particles in the context of Relativity, it is not clear what the configuration space should be now: $4N$ degrees of freedom or perhaps $1+3N$?
 If, in addition, considerations are added about mutual interactions, of possible actions at distance or of invariance under, say, the Poincaré group, it is natural to deduce the difficulties to get a theory for several relativistic particles (see, for instance, \cite{Barut,Trump,Gourgoulhon} and references therein).

 We are going to situate ourselves at a more basic level simply looking for the mere possibility of existence of a mechanical system that describes the movement of several relativistic particles (regardless of the type of existing interaction). As we will show, it turns out that this is also a source of difficulties. Let us explain this at some extent:

 A mechanical system can be defined as a configuration space $M$ (an smooth manifold), together with a (pseudo-Riemannian) metric $T_2$ and a force form $\alpha$ (a horizontal differential 1-form defined on the phase space: the tangent bundle $TM$ or its dual $T^*M$). Second Newton's law then establishes a correspondence between the form of force and a second order differential equation  $D$ (equation of motion).
 The solutions of such a second order differential  equation $D$  will be called \emph{trajectories} and are, by their very nature, parameterized curves.
 When this parameter coincides with the element of length associated with the metric $T_2$ (\emph{proper time}), we will say that the system is \emph{relativistic} (in the same way we can qualify an second order differential equation or a single trajectory as relativistic). This condition is reflected in the fact that the force form necessarily belongs to the so-called \emph{contact system} of $TM$ (see Section \ref{ss:relativisticmechanical}). Particular cases are mechanical systems with zero force (geodesic systems or free particles) and also those in which the force is of the Lorentz type whose main example is, of course,  that of a charged particle in a given electromagnetic field.

 On the other hand, there is no function that can parameterize all trajectories in a given mechanical system. Despite this, there is a class of 1-differential forms that can play that role, which we have called the \emph{class of time} (it was introduced in \cite{MecanicaMunoz}). Each 1-form of the class of time when it is restricted to any trajectory, becomes $dt$, where $t$ is the parameter of that curve. As a consequence, the integral of any one of the 1-forms in the class of time along a given trajectory measures its duration. The 
 behavior of the class of time and of the second order equations against differentiable maps leads to the conclusion that the duration remains fixed with respect to these transformations: this is the absolute character of duration \cite{TiempoMunozAlonso}. In particular, when we consider an smooth map $\varphi\colon M\to N$, if $\Gamma$ is a trajectory in $M$, then $\varphi_*\Gamma$ is also a trajectory and the durations of $\Gamma$ and $\varphi_*\Gamma$ are the same.

With the previous concepts in hand, we can now raise our problem: suppose that we now have two relativistic mechanical systems. Can these systems be part of a common `larger' mechanical system? Let us ask for even less: suppose that we now have two  relativistic trajectories. Can these trajectories be the images of a common trajectory (of some suitable mechanical system)? We will see that, in general, no. The reason is the following: consider two relativistic trajectories $\Gamma'$ and $\Gamma''$ joining points $A$ and $B$ of the same configuration space. If $\Gamma'$ and $\Gamma''$ are images by $\varphi'$ and $\varphi''$ of the same trajectory, then they will have the same duration. But that equality is incompatible with the fact that their lengths (= the proper time employed) are, as a rule, different.

The conclusion reached is net: generally, a collection of particles that move relativistically can not be the manifestation of a mechanical system, relativistic or not, at least once the definitions given in this note have been accepted.

\section{Terminology and notation}\label{notation}

In the sequel, $M$ will denote an smooth manifold. Let $\pi\colon TM\to M$ be its tangent bundle; each local chart $\{x^j\}$ on $\U\subseteq M$ defines a local chart $\{x^j,\dot x^j\}$ on $T\U\subseteq TM$ in the usual way.
\medskip

Tangent vectors or tangent fields on $TM$ that are tangent to the fibres of the projection $\pi\colon TM\to M$ will be called \emph{vertical} (so that, as derivations of the ring $\C^\infty(TM)$, annihilate the subring $\C^\infty(M)$). Locally, vertical vector fields (or vectors) looks like $ a^j\partial/\partial\dot x^j$.
\medskip

The differential 1-forms on $TM$ that (by interior product) annihilate all the vertical vectors will be called \emph{horizontal 1-forms}. For each function $f\in\C^\infty(M)$, $df$ is a horizontal 1-form and, locally, any horizontal 1-form is a linear combination of the forms $df$, with coefficients in $\C^\infty(TM)$ or, sums $b_j\,dx^j$.

\subsection{The derivation $\dot d$}
Each horizontal 1-form $\alpha$ on $TM$ defines a function $\dot\alpha$ on $TM$ by the rule $$\dot\alpha(v_a)=\langle\alpha,v_a\rangle,$$ for each $v_a\in TM$
(the symbol $\langle\alpha,v_a\rangle$ denotes the obvious pairing). In particular, for each  $f\in\C^\infty(M)$, the function $\dot{\overline {df}}\in\C^\infty(TM)$ will be denoted, for short, as $\dot f$ (this notation is compatible with the previously mentioned $\dot x^j$). For each tangent vector $v_a\in TM$ we have $\dot f(v_a)=v_a(f)$.
\medskip

The map $\dot d\colon\C^\infty(M)\to\C^\infty(TM)$, $f\mapsto\dot df:=\dot f$ is, essentially, the differential.
For each $f\in\C^\infty(M)$ we have, in coordinates,
$$\dot df=\frac{\partial f}{\partial x^j}\,\dot x^j=\left(\dot x^j\frac{\partial }{\partial x^j}\right)f.$$
Hence, the local expression for the field $\dot d$ (field on $TM$ with values  in $TM$) is
  $$\dot d=\dot x^j\frac{\partial }{\partial x^j}.$$

  By using $\dot d$ we can put also $\dot\alpha=\dot d\lrcorner\alpha$.

  \subsection{The bundle of accelerations. Second order differential equations}\label{ss:accelerations}
  A tangent vector $D_{v_a}\in T_{v_a}(TM)$ is an \emph{acceleration} when, for each $f\in\C^\infty(M)$ is $D_{v_a}f=v_af$; that is to say, when $\pi_*(D_{v_a})=v_a$ ($\pi$ is the tangent bundle projection). A field $D$, tangent to $TM$, is an \emph{second order differential equation} when the value $D_{v_a}$ at each $v_a\in TM$ is an acceleration. This means that, for each $f\in\C^\infty(M)$, we have $Df=\dot df=\dot f$. In local coordinates,
  $$D=\dot x^j\,\frac{\partial}{\partial x^j}+f^j(x,\dot x)\,\frac{\partial}{\partial\dot x^j},$$
  for certain functions $f^j\in\C^\infty(TM)$. In this way, a parameterized curve $x^j=x^j(t), \dot x^j=\dot x^j(t)$ is a solution of $D$ if and only if $x^j(t)=dx^j(t)/dt$ and $\dot x^j(t)=f^j(x(t),dx/dt)$ which can be combined into the system
  $$\frac{d^2 x^j}{dt^2}=f^j(x(t),dx(t)/dt),\quad j=1,\dots,n.$$

  Given two second order differential equations $D$, $\overline D$, the difference $D-\overline D$ is a vertical field. Therefore, second order differential equations are the sections of an affine bundle over $TM$ (the \emph{bundle of accelerations}), modeled over the vector bundle of the vertical tangent fields.
  
  \subsection{Contact forms on $TM$}
  A 1-form $\alpha$ on $TM$ is said a \emph{contact form} when annihilates  all the second order differential equations. Such an $\alpha$ also annihilates the difference of any couple of second order differential equations. Thus, $\alpha$ annihilates every vertical tangent field and, then, is a horizontal 1-form. On the other hand, for each 1-form horizontal $\alpha$ and acceleration $D_{v_a}$, we have
    $$\langle\alpha,D_{v_a}\rangle=\langle\alpha,v_a\rangle=\dot\alpha(v_a),$$
 so that $\alpha$ will be a contact form if and only if $\alpha$ is horizontal and, in addition, $\dot\alpha=\langle\alpha,\dot d\rangle=0$.

 The set of contact 1-forms is a Pfaff system on $TM$, the \emph{contact system} $\Omega$, of rank $n-1$ (if $n=\dim M$), generated out of the 0-section of $TM$, by the 1-forms
 $$\dot x^j\,dx^i-\dot x^i\, dx^j,\quad (i,j=1,\dots,n),$$
 on each coordinated open set.

  For each $v_a\in TM$, the tangent vectors on $TM$ at $v_a$ annihilated by $\Omega$ are the multiples of accelerations at $v_a$, along with the vertical vectors at $v_a$. A curve $\Gamma$ in $TM$ is a solution of the contact system if it is tangent at each point to an acceleration or a (non trivial) vertical vector. If $\Gamma$ is not vertical, then necessarily has locally the following coordinated expression
  $$x^i=x^i(t),\qquad \dot x^i=\frac{dx^i(t)}{dt},\qquad i=1,\dots,n.$$
  Therefore $\Gamma$ is the so called lifting of the parameterized curve  $\pi(\Gamma)=\{x^i=x^i(t)\}$ to the tangent bundle.

  \bigskip

\section{Absolute time}\label{s:absolutetime}

On the space $TM$ there is no function that can do the role of ``time'': a valid parameter for all the trajectories. In fact, such a parameter, say $g$, should hold $Dg=1$ for every second order equation $D$ and this is impossible (see \cite{TiempoMunozAlonso} for details).

Despite the above, there is  a certain class of differential 1-forms that can play that role:
\begin{defi}
We will say that an 1-form $\alpha$ belongs to the \emph{class of time} on an open set $\U$ of $TM$ when is horizontal and $\dot\alpha=1$ on $\U$.
\end{defi}
If $\alpha$, $\beta$ belong to the class of time on $\U$, then $\alpha-\beta\in\Omega$ on $\U$: two 1-forms in the class of time are congruent modulo the contact system. For each function $f\in\C^\infty(M)$, on the open set of $TM$ where $\dot f\ne 0$, the form $df/\dot f$ belongs to the class of time. It is derived that each point $v_a\in TM$, out of the 0-section, has a neighborhood in which there is a form in the class of time: the class of time is defined all along $TM$ except the 0-section. In fact, by using partitions of the unity can be showed that there is a global form $\ttau$ in the class of time on the complementary open set of the  0-section in $TM$.

 For each curve $\Gamma$ in $TM$, solution of the contact system, and which does not intersect the 0-section, we will call \emph{duration} of $\Gamma$ the integral $\int_\Gamma\ttau$, where $\ttau$ is an 1-form in the class of time.

 The ``absolute'' character of the duration is a consequence of the functoriality of the notion of acceleration:
 \begin{lemma}
 If $\varphi\colon M\to N$ is a morphism of smooth manifolds, and $\varphi_*\colon TM\to TN$ is the corresponding morphism between their tangent bundles, the tangent map $$\varphi_{**}\colon T(TM)\to T(TN)$$ sends accelerations in $TM$ to accelerations in $TN$.
 \end{lemma}
 \begin{proof}
 See \cite{TiempoMunozAlonso}. For each $v_a\in TM$, by differentiating the commutative diagram
 $$
 \xymatrix{ TM\ar[r]^-{\varphi_*}\ar[d]_-\pi & TN\ar[d]^-{\pi}
     & {}\ar@{}[d]^-{\displaystyle{\text{we get}}} & & T_{v_a}(TM)\ar[r]^-{\varphi_{**}}\ar[d]_-{\pi_*} & T_{\varphi_*(v_a)}(TN)\ar[d]^-{\pi_*}\\
 M\ar[r]_-\varphi & N & & & T_aM\ar[r]_-{\varphi_*} & T_{\varphi(a)}N}
 $$

 Now, if $D_{v_a}$ is an acceleration at $v_a\in TM$ (so that $\pi_*D_{v_a}=v_a$), then
 $$\pi_*(\varphi_{**}(D_{v_a}))=\varphi_*\circ \pi_*\,(D_{v_a})=\varphi_*(v_a),$$
 which entails that $\varphi_{**}(D_{v_a})$ is an acceleration at $\varphi_*(v_a)$.
  \end{proof}

\begin{cor}
$(\varphi_*)^*$ applies the contact system of $N$, $\Omega_N$, into the contact system of $M$, $\Omega_M$. $\varphi_*$ applies curves $\Gamma$ solution of $\Omega_M$ into curves $\varphi_*(\Gamma)$ solution of $\Omega_N$.
\end{cor}

\begin{cor}
For each 1-form $\tau$ in the class of time on an open set $\U$ of $TN$, the form $(\varphi_*)^*(\tau)$ is in the class of time on the open set $\varphi^{-1}(\U)$ of $TM$.
\end{cor}

\begin{cor}[``Absolute'' character of the duration]\label{absolutetime}
Let $\Gamma$ be a curve solution of $\Omega_M$, whose image $\varphi_*(\Gamma)$ does not intersect the 0-section of $TM$. The duration of $\varphi_*(\Gamma)$ is the same as the duration of $\Gamma$.
\end{cor}
\medskip

Assume that the entire Universe is the trajectory of a second order differential equation $(M, D)$. ``Forgetting degrees of freedom'' of $M$ by means of a suitable projection we will obtain a new second order differential equation on an smaller configuration space. According to the above results, it would be enough to measure the duration of time following the evolution of that smaller system. And it does not matter which we choose: we will always measure the same duration. It is in that sense that we can talk about absolute time.

\section{Mechanical systems}\label{s:mechanicalsystems}

In this section we will review the Classical Mechanics as exposed in \cite{MecanicaMunoz} and \cite{RM}.
We will continue with the notation given in Section \ref{notation}.
\medskip

Newton Mechanics is based on a link between forces and accelerations. But, for forces to produce accelerations it is necessary to have on $M$ a riemannian metric of arbitrary signature (i.e.: pseudo-Riemannian metrics, included, for instance, the Lorentzian ones).
\medskip

\subsection{Liouville form. Kinetic energy. Length element}
Let $T^*M$ be the cotangent bundle of $M$. If $\U\subset M$ is an open set coordinated by $\{x^i\}$, we define, as usual, the functions $p_i$ on $T^*\U\subset T^*M$ by the rule $p_i(\alpha_a):=\alpha_a(\partial/\partial x^i)$ for each $\alpha_a\in T^*\U$. The functions $\{x^i,p_i\}$, $i=1,\dots,n$, are  coordinates on $T^*\U$.

The \emph{Liouville form} $\theta$ is defined by
$$\theta_{\alpha_a}=\pi^*\alpha_a,\quad \text{for each}\,\,\alpha_a\in T^*M,$$
where $\pi^*$ denotes the pullback by the projection $T^*M\to M$.
Its exterior differential, $d\theta=\omega_2$, is the canonical \emph{symplectic form} in $T^*M$. Their expressions in local coordinates are $$\theta=p_i\,dx^i,\qquad
\omega_2=dp_i\wedge dx^i.$$

Given a non-singular metric $T_2$ on $M$ of arbitrary signature (pseudo-Riemannian metric), we dispose of an isomorphism $TM\simeq T^*M$ which will be tacitly assumed from now on. Thus, thanks to the metric $T_2$ on $M$, we will talk about the Liouville form $\theta$ in $TM$, which will be given by $\theta_{v_a}=v_a\lrcorner\,T_2$ (pulled-back from $M$ to $TM$).

 If, in local coordinates,  $T_2=g_{ij}\,dx^i dx^j$, the isomorphism between $TM$ and $T^*M$ is expressed by
  $$p_i=g_{ij}\dot x^j\quad\text{or}\quad \dot x^j=g^{ij}p_i,$$
   where $g^{ij}$ denote the entry $(i,j)$ of the inverse matrix of $(g_{ij})$. In particular,
  $$\theta=g_{ij}\dot x^i\,dx^j$$
 in $TM$.
 \medskip

 The function $T:=(1/2)\dot\theta=(1/2)g_{ij}\dot x^i\dot x^j$ is the \emph{kinetic energy}.
  \medskip

On the other hand, we can define intrinsically the \emph{length element} as the differential 1-form
 $$\frac{\theta}{\|\dot\theta\|}=\frac{\theta}{\sqrt{|\dot\theta|}},$$
which has a sense on the open subset of $TM$ where the kinetic energy does not vanish $\dot\theta\ne 0$.

 Let us denote by $d\tau$ the classical length element: it  is the restriction  of the length element (in the above sense) to the curve in $TM$ describing the lifting $\Gamma$ of the corresponding
parameterized curve $\gamma$ in $M$. If $\gamma$ is locally $x^i=x^i(t)$, then
$$d\tau=\frac{\theta}{\sqrt{|\dot\theta|}}\left.\right|_\Gamma=
   \frac{{g_{ij}\frac{dx^i}{dt}\frac{dx^j}{dt}}}{\sqrt{|g_{ij}\frac{dx^i}{dt}\frac{dx^j}{dt}}|}\,dt=
\sqrt{\left|g_{ij}\frac{dx^i}{dt}\frac{dx^j}{dt}\right|}\,dt,$$
which matches with the usual coordinate expression.

 When we have a curve $\C\subseteq M$, each of its possible parameterizations $\gamma$ defines a (generally) different lifting $\Gamma$ to $TM$.
 It turns out that integral of $\theta/\sqrt{|\dot\theta|}$ along $\Gamma$ does not
depend on the parametrization, and this is the reason for which it
has sense  to talking about the length of a curve $\C$.
\medskip

\subsection{The Fundamental Lema. Mechanical systems}

 The entirety of the Mechanics rests on the following
\begin{lemma}[Fundamental Lemma of Classical Mechanics]
The metric $T_2$ establishes a univocal correspondence between second order differential equations $D$ on $M$ and horizontal 1-forms $\alpha$ in $TM$, by means of the following equation
\begin{equation}\label{Newton}
D\lrcorner\,\omega_2+dT+\alpha=0.
\end{equation}
The tangent fields $u$ on $M$ that are intermediate integrals of $D$ are precisely those holding
\begin{equation}\label{NewtonIntermedia}
u\lrcorner\, d(u\lrcorner\,T_2)+dT(u)+u^*\alpha=0, 
\end{equation}
where $u^*\alpha$ is the pull-back of $\alpha$ by means of the section $u\colon M\to TM$ and $T(u)$ is the function $T$ specialized to $u$.
\end{lemma}
\begin{proof}
See \cite{MecanicaMunoz,RM}.
\end{proof}

The local expression for the second order differential equation is
$$D=\dot x^j\frac{\partial}{\partial x^j}+\ddot x^j\frac{\partial}{\partial\dot x^j},\qquad \text{where}\quad \ddot x^i:=D\dot x^i$$
is
given by
\begin{equation}\label{Newtoncoordenadas2}
\ddot x^j+\Gamma_{k\ell}^j\,\dot x^k\dot x^\ell+\alpha^j=0,
\end{equation}
where $\Gamma_{k\ell}^j$ denotes the Christoffel symbols of the metric
(raising and lowering of indexes will be associated to $T_2$: for instance, $\alpha^j=g^{ij}\alpha_i$).\medskip


  \begin{defi}\label{d:mechanicalsystem}
  A \emph{classical mechanical system} is a manifold $M$ (the \emph{configuration space}) endowed with a pseudo-Riemannian metric $T_2$ and a horizontal 1-form $\alpha$, the \emph{form of work} or \emph{form of force}. The second order differential equation $D$ corresponding with $\alpha$ by (\ref{Newton}) is the \emph{differential equation of the motion} of the system $(M,T_2,\alpha)$ and (\ref{Newton}) is the \emph{Newton equation}. The solution curves of $D$ will be called \emph{trajectories} of the mechanical system.
  \end{defi}

  For $\alpha=0$ we have the \emph{system free of forces} or \emph{geodesic system}, whose equation of motion $D_G$ is the \emph{geodesic field}: ${D_G}\lrcorner\,\omega_2+dT=0$. $D_G$ is the hamiltonian field corresponding to the function $T$ by means of the symplectic form $\omega_2$ associated with the given metric.

  The geodesic field provides an origin for the affine bundle of the second order differential equations (see paragraph \ref{ss:accelerations}). For each second order equation $D$, the difference $D-D_G=V$ is a vertical field, the \emph{force} of the system $(M,T_2,\alpha)$ whose equation of the motion is $D$.
  The force $V$ and the form of force $\alpha$ are related  by
  $V\lrcorner\,\omega_2+\alpha=0$.

\subsection{Conservative systems}\label{subsectionconservatives}

When $\alpha=dU$, for a function $U\in\C^\infty(M)$, the mechanical system $(M,T_2,\alpha)$ is said to be \emph{conservative}.
In this case, equations (\ref{Newtoncoordenadas2}) are
\begin{equation}\label{Newtoncoordenadas3}
g_{i j}\ddot x^j+\Gamma_{k\ell,i}\,\dot x^k\dot x^\ell+\frac{\partial U}{\partial x^i}=0.
\end{equation}

The sum $H:=T+U$ is called \emph{Hamiltonian} of the system. Newton equation  (\ref{Newton}) becomes in this case,
\begin{equation}\label{Newtonconservativo}
D\lrcorner\,\omega_2+dH=0.
\end{equation}

When (\ref{Newtonconservativo}) is written in coordinates $(x,p)$ of $T^*M$, we get the system of Hamilton \emph{canonical equations}. The specialization of $\omega_2$ into each hypersurface $H=\textrm{const.}$ of $T^*M$,  has as radical the field $D$ (and its multiples), as it is derived from (\ref{Newtonconservativo}). From the classical argument based in the Stokes theorem it results, then, the \emph{Maupertuis Principle}: the curves in $M$ with given end points that, lifted to $TM$, remain within the same hypersurface $H=\textrm{const.}$, give values for $\int\theta$ which are extremal just for the trajectory of the system.

The equation (\ref{NewtonIntermedia}) for the intermediate integrals of $D$ is, in this case:
\begin{equation}\label{NewtonconservativoIntermedia}
u\lrcorner\,d(u\lrcorner\,T_2)+dH(u)=0.
\end{equation}

In particular, if $u$ is a Lagrangian submanifold of $TM$, this is to say, if $d(u\lrcorner\,T_2)=d\theta|_u=0$, the equation (\ref{NewtonIntermedia}) is
$$dH(u)=0\quad\text{or}\quad H(u)=\textrm{const.}$$
This is the \emph{Hamilton-Jacobi equation}, when is translated to the function $S$ of which $u$ is the gradient:
\begin{equation}\label{Hamilton-Jacobi}
H(\textrm{grad}\,S)=\textrm{const.}\,\,\text{in $TM$, or}\quad H(dS)=\textrm{const.} \,\,\text{in $T^*M$,}
\end{equation}
or, in coordinates,
$$\frac 12\,g^{jk}\,\frac{\partial S}{\partial x^j}\frac{\partial S}{\partial x^k}+U=0.$$
\medskip

\section{Relativistic mechanical systems}\label{ss:relativisticmechanical}

In Special Relativity, the kinetic energy or, equivalently, the length of the vector velocity along a given trajectory stays constant. So, it is natural to give the following
\begin{defi}\label{DefRelativista}
A \emph{relativistic field} on $(M,T_2)$ is a second order
differential equation $D$ such that $DT=0$, where $T=(1/2)\dot\theta$ is the
kinetic energy associated with $T_2$. In the same way, the mechanical system associated with such a  field $D$ will be said to be a \emph{relativistic mechanical system}.
\end{defi}

The above condition, slightly unexpectedly,  does not depend on the metric (see \cite{RelatividadMunoz,RM}). In fact, it holds
\begin{thm}\label{criterio}
A second order differential equation $D$ on $(M,T_2)$ is a
relativistic field if and only if the work form $\alpha$
associated by virtue of the Newton-Lagrange law
($i_D\omega_2+dT+\alpha=0$) belongs to the contact system $\Omega$
of $TM$. That is equivalent to $\dot\alpha=0$.
\end{thm}

 \begin{obs}
 \emph{Relativistic correction}: It is always possible to make a change in the force in the direction of the trajectories so that a given mechanical system becomes relativistic \cite{RM}: it is enough change the force form $\alpha$ by $\widetilde\alpha:=\alpha-(\dot\alpha/\dot\theta)\theta$.
 \end{obs}

 Relativistic fields are the natural generalization of the geodesic
 field $D_G$, the one corresponding with $\alpha=0$ (in the
 contact system).
 \medskip

 Another important example of relativistic system is given by  \emph{Lorentz-type forces}: Let $F$ be a differential 2-form on a manifold $M$ and
let us take
$$\alpha={\dot d}\lrcorner F;$$
in coordinates, $$\alpha=F_{ij}(\dot x^idx^j-\dot x^jdx^i)$$ if $F=F_{ij}dx^i\wedge dx^j$. Since $\dot\alpha=\dot d\lrcorner\alpha=\dot d\lrcorner\dot d\lrcorner F=0$, the mechanical system $(M,T_2,{\dot d}\lrcorner F)$ is relativistic for arbitrary metric $T_2$.
\medskip

On the contrary, conservative systems (except the geodesic case) can never be relativistic: for no non-constant function $U\in\mathcal C^\infty(M)$, the 1-form $\alpha=dU$ will belong to the contact system because $\dot\alpha=\dot U=\dot x^j\partial U/\partial x^j\ne 0$.
\medskip

Let us consider a relativistic mechanical system $(M,T_2,\alpha)$ and so, along each of its solutions, $\dot\theta=2T$ keeps constant. In addition length element is
$d\tau:=\theta/\sqrt{|\dot\theta|}$ and a representative of the time class is $\theta/\dot\theta$ (assumed that $\dot\theta\ne 0$).
As a consequence, along a solution parameterized by $t$, the element of length is given by:
$$d\tau=\frac{\theta}{\sqrt{|\dot\theta|}}=\sqrt{|\dot\theta|}\,\,\frac{\theta}{\dot\theta}=\sqrt{|\dot\theta|}\,dt=kdt,$$
where $k$ is the value of $\sqrt{|\dot\theta|}=\sqrt{|g_{ij}\dot x^i\dot x^j|}$ constant along the trajectory.

In the context of Relativity Theory, length of a trajectory is known as \emph{proper time}, so that the above discussion can be stated as follows.
\begin{prop}
The trajectories of a relativistic mechanical system are parameterized by a  multiple of the proper time constant along each trajectory.
\end{prop}

\medskip

\subsection{Strictly relativistic trajectories}\label{ss:strictlyrelativistic}

Now we will select from among the possible trajectories of a relativistic mechanical system a particular type of them.

\begin{defi}\label{strictrelativistic}
The trajectories $\Gamma$ of a relativistic mechanical system $(M,T_2,\alpha)$ which are parameterized by the proper time $\tau$ (so that, $d\tau=dt$ or, equivalently, $|\dot\theta|=1$) will be called \emph{strictly relativistic trajectories}.
\end{defi}

For a given parameterized geodesic (solution of the relativistic mechanical system $(M,T_2,0)$), there is a re-parametrization with which it becomes a strictly relativistic trajectory. However, this new parameterized curve is different from the original one (their lifts to the tangent bundle $TM$ are not the same, although they travel along the same points of $M$). In essence, that is due to the second derivatives depending quadratically on the velocities. In the same way, for a relativistic mechanical system $(M,T_2,\alpha)$, where $\alpha$ depends quadratically on the velocities, we can transform every solution in a strictly relativistic solution by means of a re-parametrization: is is sufficient to replace the parameter $t$ by $kt$, where $k$ is the constant value of $\sqrt{|\dot\theta|}$ on the given parameterized path. However, the above will not be possible, for example, in the case of Lorentz-type forces which depend linearly on the velocities.

\section{Nonexistence of configuration space for a couple of relativistic particles}\label{nonexistence}

Let us consider two relativistic mechanical systems ${\s}'=(M,T'_2,\alpha')$ and ${{\s}''}=(M, T''_2,\alpha'')$ defined on the same manifold $M$. The following question arises: Can we consider ${\s}'$ and ${\s}''$  as part of a larger system that contains both?

In order to give an answer, let us fix two strictly relativistic solutions $\Gamma'$, $\Gamma''$, one of each on of these system ${\s}'$ and $\s''$. Moreover, assume that both of them connect the initial point $A\in M$ with the final point  $B\in M$. Now, we wonder if there is a mechanical system $\widehat{\s}=(\widehat M,\widehat T_2,\widehat\alpha)$, smooth maps $\varphi'\colon\widehat M\to M'$, $\varphi''\colon\widehat M\to M''$ and a solution $\widehat\Gamma$ of $\widehat S$ such that its images through $\varphi'_*$ and $\varphi''_*$ are $\Gamma'$ and $\Gamma''$ respectively?

According to the results in Section \ref{s:absolutetime}, the durations of $\Gamma'$ and $\Gamma''$ match that of $\widehat\Gamma$:
$$\text{Duration of $\Gamma'$}=\text{Duration of $\widehat\Gamma$}=\text{Duration of $\Gamma''$}$$
but, in addition, as $\Gamma'$ and $\Gamma''$ are strictly relativistic, we have
$$\text{Duration of $\Gamma'$}=\text{Proper time elapsed in $\Gamma'$}$$
$$\text{Duration of $\Gamma''$}=\text{Proper time elapsed in $\Gamma''$}.$$

As a consequence, we arrive to the equality:
$$\text{Proper time elapsed in $\Gamma'$}=\text{Proper time elapsed in $\Gamma''$}.$$

Finally, from it emerge a contradiction because it is easy to find two strictly relativistic paths with different ``proper durations'' (twins paradox, see Appendix).

As a consequence, we respond negatively to the question proposed at the beginning of this section.
\medskip

It could be argued that the consideration of relativistic trajectories, but not necessarily strictly relativistic,  improve the situation. However, this also leads to difficulties: let $\Gamma_C$ be the geodesic trajectory starting at the point $A$ and cutting  $\Gamma'$ at some point $C$. In that case, the contradiction found above could be avoided by choosing a suitable velocity $k_C=\sqrt{|\dot\theta|}\left.\phantom{\dot|}\right|_{\Gamma_C}$ for $\Gamma_C$ (this is always possible by multiplying the parameter by a appropriate constant factor; see the Appendix for an example). That seems to clash with considerations about the isotropy of space: Why should the velocities of free particles vary in magnitude depending on their direction? It seems something artificial and unjustified.

\section*{Appendix}
For questions of completeness, we now give a concrete example of two strictly relativistic trajectories joining two points of a space (or spacetime) with different proper durations.
In fact, the model we present below is not completely realistic because it is well known that the presence of masses and charges modifies both the gravity and the given electromagnetic field. That is why, although mathematically rigorous, the model  must be interpreted as a description of the motion of test particles: the conclusions of the following do not change (except for constant scale factors) if the masses and charges were `infinitesimal' and the ratio charge/mass is equal to 1, for instance. To put it another way, the example deals with a limiting case.
\medskip

Let $(M=\R^4,T_2)$ be the Minkowski spacetime manifold coordinated by $(x^0,x^1,x^2,x^3)$ and metric
$T_2=(dx^0)^2-(dx^1)^2-(dx^2)^2-(dx^3)^2$ (so that $c=1$); let us consider two particles of mass 1, one of which has a charge $e=1$; consider also the presence of a constant unity magnetic field parallel to axis $OX^3$, say $\alpha=\dot x^1dx^2-\dot x^2dx^1$. (To simplify, we will identify a trajectory in the tangent bundle with its projection to the configuration manifold $M$).
\smallskip

\noindent{\bf Charged particle.} Under the appropriate initial conditions (initial velocity in the plane $x^0x^2$, etc.), the charged  particle describe the trajectory $\Gamma'$: $$x^0=\lambda t,\quad x^1=\eta\,\textrm{cos}(t),\quad x^2=\eta\,\textrm{sin}(t),\quad x^3=0,\qquad t\in[0,\pi],$$
for given constants $\eta,\lambda$. So that $\Gamma'$ joints the points $A=\Gamma'(0)=(0,\eta,0,0)$ and $B=\Gamma'(\pi)=(\lambda\pi,-\eta,0,0)$.  The proper time along $\Gamma'$ is $d\tau=\sqrt{\lambda^2-\eta^2}\,dt$. Therefore, in the case $\sqrt{\lambda^2-\eta^2}=1$, 
 the trajectory $\Gamma'$ is parameterized by the proper time (and then, it is strictly relativistic).
\smallskip

\noindent{\bf Neutral particle.} On the other hand, the not charged particle travels along  straight lines; for instance, let us take $\Gamma''$  as follows:
$$x^0= \mu t,\quad x^1=-\frac{2\mu \eta}{\pi\lambda}\, t+\eta,\quad x^2=0,\quad x^3=0,\qquad t\in[0,\lambda\pi/\mu],$$
which also joints $A=\Gamma''(0)$ and $B=\Gamma''(\lambda\pi/\mu)$; the path $\Gamma''$ will be
parameterized by the proper time if $\mu$ is chosen such that  $\mu\sqrt{1- 4\eta^2/(\lambda^2\pi^2)}=1$ because along $\Gamma''$ we have
$d\tau=\sqrt{\mu^2-\mu^2 4\eta^2/(\lambda^2\pi^2)}\,dt$.
\medskip

The duration of $\Gamma'$ is $\pi$ and the duration of $\Gamma''$ is $\lambda\pi/\mu$. These two durations are different if $\mu\ne\lambda$, which is true (except, at most, for some exceptional values of $\eta$) because the above choices for the constants give
$$\lambda=\sqrt{1+\eta^2},\qquad \mu=\frac 1{\sqrt{1-\frac{4\eta^2}{\pi^2(1+\eta^2)}}}.$$
\bigskip

\noindent{\bf Trajectories not strictly relativistic.} In order to illustrate the possibility of relaxing the above conditions
we can impose equality of duration but allow different velocities (so that they will no longer be strictly relativistic trajectories but only relativistic ones). For this, it will be enough to modify the allowed velocity in one of the two trajectories, leaving the velocity of the other fixed. Let us consider, for example, a straight line $\Gamma_C$ starting from $A$ to reach a variable point $C=\Gamma'(s)$ and using the same duration as $\Gamma'$, $t\in[0,s]$ (which is $s$). A little computation gives a unique possibility for $\Gamma_C$:
$$x^0= \sqrt{1+\eta^2}\,\, t,\quad x^1=\eta\frac{\cos(s)-1}{s}\, t+\eta,\quad x^2=\eta\frac{\sin(s)}{s}\, t,\quad x^3=0,\qquad t\in[0,s],$$
and the magnitude of its velocity is
$$k_C=\left.\sqrt{|\dot\theta|}\,\right|_{\Gamma_C}=\sqrt{1+\eta^2\left(1-2\frac{1-\cos(s)}{s^2}\right)}.$$
We see that the necessary choice of speed so that there is no contradiction between duration and proper time depends on the point $C=\Gamma'(s)$. In other words, if we want to keep consistency between durations, it is necessary that the initial conditions for the particles are adjusted depending on the direction. It seems an unreasonable demand.
\bigskip

\subsection*{Concluding remarks.}
The arguments given in this note seem to prevent the description of relativistic $N$-particle systems by classical mechanical systems of arbitrary type (it does not matter what type of metric is assumed nor the dimension of the configuration space, etc.). The contradiction is due to a ``twins paradox effect'' that appears by the coexistence of absolute time (something inevitable in mechanical systems) and the encounters between particles. The possibility that such incompatibility can be overcome with the addition of interactions seems to be discarded: in the example of the Appendix section, mass and charge can be assumed infinitesimal, provided that the charge/mass ratio remains finite. It is worth noting also that the arguments do not depend on whether we deal specifically with the Minkowski metric: the difficulty appears for any pseudo-Riemannian metric. In the framework presented here, the road seems to be closed, and we do not know the possible way out, so other paths must be explored. Perhaps a reasonable way is to consider mechanical systems, say for two particles, each of which behaves relativistically only when they are sufficiently far apart: so to speak, they would have an asymptotically relativistic behavior.
\bigskip


\end{document}